\documentclass[11pt,letterpaper]{article}

\usepackage{typearea}
\typearea{15}
\usepackage{multirow}
\usepackage{booktabs}
\usepackage{theorem,latexsym,graphicx}
\usepackage{amsmath,amssymb,enumerate}
\usepackage{xspace}
\usepackage{bm}
\usepackage{ifpdf}
\usepackage{shadow,shadethm,color}
\usepackage[compact]{titlesec}
\usepackage{algorithm}
\usepackage{subfigure}
\usepackage{verbatim}
\usepackage{times}
\usepackage{paralist}
\usepackage[noend]{algpseudocode}
\usepackage[bottom]{footmisc}

\definecolor{Darkblue}{rgb}{0,0,0.4}
\definecolor{Brown}{cmyk}{0,0.81,1.,0.60}
\definecolor{Purple}{cmyk}{0.45,0.86,0,0}
\newcommand{\mydriver}{hypertex}
\ifpdf
 \renewcommand{\mydriver}{pdftex}
\fi
\usepackage[breaklinks,\mydriver]{hyperref}
\hypersetup{colorlinks=true,
            citebordercolor={.6 .6 .6},linkbordercolor={.6 .6 .6},%
citecolor=Darkblue,urlcolor=black,linkcolor=Darkblue,pagecolor=black}
\newcommand{\lref}[2][]{\hyperref[#2]{#1~\ref*{#2}}}

\newtheorem{theorem}{Theorem}[section]

\newtheorem{lemma}[theorem]{Lemma}
\newshadetheorem{lemmashaded}[theorem]{Lemma}

\newtheorem{example}[theorem]{Example}

\numberwithin{algorithm}{section}

\newenvironment{proof}{

\noindent{\bf Proof:}}
{\hfill$\blacksquare$

}

\newcommand{\junk}[1]{}
\newcommand{\ignore}[1]{}

\newcommand{\R}[0]{{\ensuremath{\mathbb{R}}}}

\newcommand{\e}{\varepsilon}
\newcommand{\eps}{\varepsilon}

\newcommand{\OPT}{\ensuremath{{\sf Opt}}}

\newcommand{\E}{\mathbb{E}}

\newcommand{\U}{\mathcal{U}}

\newcounter{mynotes}
\newcommand{\mmnote}[1]{}

\newcommand{\qedsymb}{\hfill{\rule{2mm}{2mm}}}
\renewenvironment{proof}{\begin{trivlist} \item[\hspace{\labelsep}{\bf
\noindent Proof.\/}] }{\qedsymb\end{trivlist}}%

\newcommand{\initOneLiners}{%
    \setlength{\itemsep}{0pt}
    \setlength{\parsep }{0pt}
    \setlength{\topsep }{0pt}
}

\newcommand{\lpGLoadBal}{$\ell_p$-\textsc{Generalized Load-balance}\xspace}
\newcommand{\GLB}{\textsc{GLB}$_p$\xspace}

\newcommand{\greedy}{\textsc{Greedy}\xspace}
\newcommand{\ultimate}{\textsc{SimultaneousLB}\xspace}
\newcommand{\greedywr}{\textsc{GreedyWR}\xspace}
\newcommand{\smoothgreed}{\textsc{SmoothGreedy}\xspace}
\newcommand{\smoothbaseline}{\textsc{SmoothBaseline}\xspace}
\newcommand{\algo}{\textsf{Algo}\xspace}

\newcommand{\bA}{\mathbf{A}}
\newcommand{\bx}{\mathbf{x}}
\newcommand{\bS}{\mathbf{S}}

\newcommand{\bg}{\mathbf{g}}
\newcommand{\btau}{\bm{\tau}}
\newcommand{\bell}{\bm{\ell}}
\newcommand{\bY}{\mathbf{Y}}
\newcommand{\bZ}{\mathbf{Z}}

\newcommand{\red}[1]{\textcolor{red}{#1}}

\newcommand{\ip}[2]{\langle #1, #2 \rangle\xspace}
\newcommand{\ones}{\mathbf{1}}
\newcommand{\holder}{H\"older\xspace}

\newcommand{\OPTfixed}{\mathsf{OptFixed}\xspace}

\newcommand{\bj}{\mathbf{j}}
\newcommand{\bu}{\mathbf{u}}
\newcommand{\bX}{\mathbf{X}}

\newcommand{\short}[2]{#1}
\newcommand{\remove}[1]{}

\begin{document}

\title{
Online and Random-order Load Balancing Simultaneously}
\author{Marco Molinaro\footnote{Email: mmolinaro@inf.puc-rio.br} \\ PUC-Rio, Brazil}
\date{}

\maketitle

\begin{abstract}
	We consider the problem of online load balancing under $\ell_p$-norms: sequential jobs need to be assigned to one of the machines and the goal is to minimize the $\ell_p$-norm of the machine loads. This generalizes the classical problem of scheduling for makespan minimization (case $\ell_\infty$) and has been thoroughly studied. However, despite the recent push for beyond worst-case analyses, no such results are known for this problem.
	
	In this paper we provide algorithms with \textbf{simultaneous} guarantees for the worst-case model as well as for the random-order (i.e. secretary) model, where an arbitrary set of jobs comes in random order. First, we show that the greedy algorithm (with restart), known to have optimal $O(p)$ worst-case guarantee, also has a (typically) improved random-order guarantee.	However, the behavior of this algorithm in the random-order model degrades with $p$. We then propose algorithm \ultimate that has \textbf{simultaneously optimal guarantees} (within constants) in both worst-case and random-order models. In particular, the random-order guarantee of \ultimate improves as $p$ increases.
	
	One of the main components is a new algorithm with improved regret for Online Linear Optimization (OLO) over the non-negative vectors in the $\ell_q$ ball. Interestingly, this OLO algorithm is also used to prove a purely probabilistic inequality that controls the correlations arising in the random-order model, a common source of difficulty for the analysis. Another important component used in both \ultimate and our OLO algorithm is a smoothing of the $\ell_p$-norm that may be of independent interest. This smoothness property allows us to see algorithm \ultimate as essentially a greedy one in the worst-case model and as a primal-dual one in the random-order model, which is instrumental for its simultaneous guarantees. 
\end{abstract}

	\pagenumbering{gobble}
	\newpage
	

	\pagenumbering{arabic}
	\section{Introduction}

		We study the following classical online \lpGLoadBal (\GLB) problem: There are $m$ machines, and $n$ jobs come one-by-one. Each job can be processed in the machines in $k$ different ways, so the $t$-th job has an $m \times k$ matrix $A^t$ with entries in $[0,1]$ whose column $j$ gives the load $(A^t_{1j}, A^t_{2j}, \ldots, A^t_{mj})$ the machines incur if the job is processed with option $j$. When the $t$-th job arrives, the algorithm needs to select a processing option for it (namely a vector $\bar{x}^t \in \{0,1\}^k$ with exactly one 1) based only on the jobs seen thus far, and the goal is to minimize the $\ell_p$-norm of the total load incurred in the machines $\| \sum_{t = 1}^n A^t \bar{x}^t \|_p$, where $\|u\|_p := (\sum_i u_i^p)^{1/p}$. The performance of the algorithm is compared against the offline optimal solution 
	$\OPT := \min \|\sum_t A^t x^t\|_p$.

		This generalizes the fundamental problem of scheduling on unrelated machines to minimize makespan, which corresponds to the case $\ell_\infty$ (and diagonal matrices $A^t$'s). The generalization to the $\ell_p$-norm has been studied since the 70's \cite{chandra,cody}, since in some applications they better capture how well-balanced an allocation is~\cite{awerbuch}. \mmnote{Citation from vector load bal of Sungjim: the L2 norm of machine loads is suitable for
disk storage [17], [20] while the Lr norm for r 
between 2 and 3 is used for modeling energy consumption [3], [38], [44].}
		
		Optimal (within constants) guarantees for this problem are well-known. Awerbuch et al. \cite{awerbuch} showed that the greedy algorithm that chooses the processing option that least increases the $\ell_p$-load has $O(p)$-competitive ratio, namely the algorithm's load is at most $O(p) \cdot \OPT$. They also provided the following matching lower bound (this is a slightly more general statement, but the proof is basically the same; we present it in Appendix \ref{app:LBWC} for completeness).
		
		\begin{theorem}[Extension of \cite{awerbuch}] \label{thm:LBWC}
		Consider \GLB in the worst-case model. Then for any positive integer $M$, there is an instance where $\OPT = M m^{1/p}$ but any (possibly randomized) online algorithm has expected load at least $\frac{1}{2^{2+1/p}} \cdot p M m^{1/p}$.
		\end{theorem}
		For the makespan case of $\ell_\infty$, one can apply the greedy algorithm with the $\ell_{\log m}$-norm to obtain an optimal $O(\log m)$-approximation (see also \cite{aspnes}); this uses the fact that $\ell_p$ for large $p$ approximates $\ell_\infty$, see Section \ref{sec:road}.
Special cases with improved guarantees~\cite{uniformMachines,allNorm,charikarLoadBal,CaragiannisRestricted,preemptiveTwoUniform}, as well as more general version of \GLB~\cite{vectorSched,convexOnlineAzar,convexOnlineAnupam}, have also been studied.  
		
		However, despite all these results, the \GLB problem has been mostly overlooked in \textbf{non-worse-case} models. Such models have received considerable attention recently, since avoiding worst-case instances often allows one to give algorithms with stronger guarantees that can be more representative of the behavior found in practice.
	A popular non-worst-case model is the \textbf{random-order} (i.e. secretary) model, where in this context the set of jobs is arbitrary but they come one-by-one in uniformly random order (see \cite{meyerson,BabaioffSurvey,DevanurHayes09,kesselheim} for a few examples).
	
	
	Even better are algorithms that have \textbf{simultaneously} a worst-case guarantee and an improved random-order guarantee. There only seems to be a few examples of such strong guarantees for different problems in the literature, most of them obtained quite recently~\cite{meyerson,adSim,welfareSim}.
		
		Our main contribution is to provide algorithms for the \GLB problem that attain simultaneously optimal worst-case competitive ratio as well as stronger guarantees in the random-order model (see Table~\ref{tab:results}). In fact, we provide algorithm \ultimate that has \textbf{optimal guarantees} (within constants) for both worst-case and random-order models. These are also the first random-order guarantees for this general problem (such results were not known even for the \emph{non-generalized} load balancing problem where the matrices $A^t$'s are diagonal).
		


		\subsection{Our results}
		
\renewcommand{\arraystretch}{1.15}
\setlength{\tabcolsep}{.3em}

\begin{table}[b]
	\footnotesize
	\center
	\begin{tabular}{c|l|l|l|}
		\cline{2-4}
		 & \multicolumn{1}{|c|}{Algorithm} & \multicolumn{1}{|c|}{Worst-case} & \multicolumn{1}{|c|}{Random-order}\\
		 \hline
		 \multicolumn{1}{|c|}{\multirow{2}{*}{$p \in [2, \infty)$}} & \greedywr & $O(p)$ ~{\scriptsize\cite{awerbuch}}& $\pmb{(1+\e)\OPT + O(\frac{p m^{1-1/p}}{\e})}$\\
		 \multicolumn{1}{|c|}{} & \textbf{\ultimate} & $\pmb{O(p)}$ & $\pmb{(1+\e) \OPT + O(\frac{p(m^{1/p} - 1)}{\e})}$\\
		 \hline
		 \multicolumn{1}{|c|}{\multirow{3}{*}{$p = \infty$}} & \greedywr, using $p\approx \frac{\log m}{\e}$ & $O(\frac{\log m}{\e})$ ~{\scriptsize\cite{awerbuch}} & $\pmb{(1+\e) \OPT + O(\frac{m \log m}{\e^2})}$\\
		 \multicolumn{1}{|c|}{} & \textsc{ExpertLB}~{\scriptsize\cite{guptaMolinaro}} & -- & $(1+\eps)\OPT + O(\frac{\log m}{\e})$ \hfill{\scriptsize\cite{guptaMolinaro}} \\ 
		 \multicolumn{1}{|c|}{} & \textbf{\ultimate, using} $\pmb{p\approx \frac{\log m}{\e}}$ & $\pmb{O(\frac{\log m}{\e})}$ & $\pmb{(1+\e) \OPT + O(\frac{\log m}{\e})}$\\
		 \hline 
	\end{tabular}
	
	\caption{Worst-case competitive ratio and random-order guarantees for \lpGLoadBal. New results are shown in bold.}
	\label{tab:results}
\end{table}
		
		
			
			
		\paragraph{Simultaneous guarantee for the greedy algorithm.} 
		Our first result shows that a small modification of the greedy algorithm, namely restarting it at time $n/2$, maintains optimal $O(p)$-approximation in the worst-case model while having improved approximation guarantee for the random-order model.
		
		\begin{theorem} \label{thm:greedy}
			For all $p \in [2, \infty)$, the greedy algorithm with restart \greedywr has the following guarantees:
			\vspace{-5pt}
			\begin{enumerate}
				\item[(a)] In the worst-case model is $O(p)$-competitive 
				\vspace{-10pt}
				\item[(b)] In the random-order model has expected load at most $(1+4\e) \OPT + \frac{(3p + 1)\, m^{1-1/p}}{\e}$ for all $\e \in (0,1]$.\footnote{Recall that we have assume all loads to be in the interval $[0,1]$; in general, the additive term in this expression scales with max load.} 
			\end{enumerate}
			Moreover, for $p= \infty$, \greedywr with $p=\Theta(\frac{\log m}{\e})$ has worst-case competitive ratio $O(\frac{\log m}{\e})$ and random-order guarantee $(1+\e) \OPT + O(\frac{m \log m}{\e^2})$.
		\end{theorem}
			Note that the lower bound from Theorem \ref{thm:LBWC} shows that in the worst-case model no algorithm can have guarantee of the form $cst \cdot \OPT + \alpha$ with $\alpha$ depending only on $m$ and $p$, and hence the random-order guarantee of Theorem \ref{thm:greedy} does not hold in the worst-case model. Moreover, typically $\OPT$ grows with the number of jobs $n$; in this case, the guarantee becomes $(1+\e) \OPT + o(\OPT)$, asymptotically giving arbitrarily close approximations, a big improvement over the best possible $O(p)\,\OPT$ worst-case guarantee.
		
		A main ingredient for proving the random-order guarantee is the optimal modulus of strong smoothness of $\|.\|_p^2$ proved recently in the context of inequalities for the $\ell_p$-norm of random vectors~\cite{AMS}.	Also, as in \cite{guptaMolinaro}, restarting the algorithm \short{}{at time $\frac{n}{2}$} reduces the correlations that arise in the random-order model: at each step, the current state now depends on at most $\frac{n}{2}-1$ jobs, so the next \short{}{random} job \short{}{still} has ``enough randomness'' for the analysis to go through.

		\paragraph{Improved simultaneous guarantee and Online Linear Optimization.} While the above algorithm typically \emph{asymptotically} gives arbitrarily close approximations in the random-order model, notice the guarantee degrades as $p$ increases, as it happens in the worst-case model. The following simple extreme example illustrates this.
		
		\begin{example}
			Consider an instance for $p= \infty$ with $m$ machines and $m$ jobs, with 2 processing options each, where job $i$'s processing options have load vectors $(1-\e, 1-\e, \ldots, 1-\e)$, for $\e \in (0,1)$, and $e^i$ (the $i$th canonical vector). It is easy to see that regardless of the order of the jobs, \greedywr always chooses processing option $(1-\e, \ldots, 1-\e)$, incurring total load $m (1-\e)$. On the other hand, $\OPT = 1$. This gives a $\Omega(m)$ additive/multiplicative gap even in the random-order model. (This example still holds for finite $p \gg m, \frac{1}{\e}$.)
		\end{example}
		
		However, we provide a new algorithm, \ultimate, that has \textbf{simultaneously optimal guarantees} (within constants) in both worst-case and random order models. In particular, its random-order guarantee improves with $p$.
		
					
		\begin{theorem} \label{thm:ultimate}
			For all $p \in [2, \infty)$ algorithm \ultimate has the following guarantees:
			\vspace{-5pt}
			\begin{enumerate}
				\item[(a)] In the worst-case model is $O(p)$-competitive 
				\vspace{-10pt}
				\item[(b)] In the random-order model has expected load at most $(1 + 4\e) (\OPT + \frac{6 p (m^{1/p} - 1)}{\e})$ for any $\e \in (0, 1]$.
			\end{enumerate}
				Moreover, in the case $p=\infty$, \ultimate with $p = \Theta(\frac{\log m}{\e})$ has worst-case competitive ratio $O(\frac{\log m}{\e}) \OPT$ and random-order guarantee $(1+\e) \OPT + O(\frac{\log m}{\e})$.
		\end{theorem}
				
		The function $p (m^{1/p} - 1)$ is decreasing in $p$, hence the random-order bound of algorithm \ultimate is always better (within constants) than that of \greedywr. Moreover, this function converges to $\ln m$ as $p$ goes to infinity~\cite{wikiLog}. Thus, this guarantee matches the only known result for \GLB in the random-order model, the $(1+\e) \OPT + O(\frac{\log m}{\e})$ bound given for the special case $p = \infty$ in~\cite{guptaMolinaro}. 
		Moreover, setting in hindsight the approximately optimal value $\e = \sqrt{\frac{p (m^{1/p} - 1)}{\OPT}}$ shows that \ultimate's solution has load at most $\OPT + O(\sqrt{\OPT \cdot p m^{1/p}})$. The following result, whose proof is presented in Appendix \ref{app:lb}, shows that this guarantee cannot be significantly improved. 
		
	\begin{theorem} \label{thm:lb}
		For every even $p \ge 2$, there is an instance of \GLB with $m = 2^p$ and $\OPT = \frac{p m^{1/p}}{2}$ such that 
		any algorithm incurs expected total load at least $\OPT + cst \sqrt{\OPT \cdot p m^{1/p}}$ for constant $cst = 1/(100 \sqrt{2})$.
	\end{theorem}
		
			The main idea behind algorithm \ultimate, or more precisely its precursor \smoothgreed, is that we can see it simultaneously as approximately \textbf{greedy} and as an approximately \textbf{primal-dual} algorithm; from the ``greedy'' part we get the worst-case guarantee, and from the ``primal-dual'' part the random-order guarantee. 
			
			Moreover, in the approximately primal-dual view, the dual variables are set according to a new algorithm for \textbf{Online Linear Optimization} (OLO) over the non-negative vectors in the $\ell_q$-ball (the dual of the $\ell_p$-ball). In this game, in each round the player needs to choose a non-negative vector vector $v^t \in \R^m$ with $\ell_q$-norm at most 1, and then the adversary chooses a non-negative vector $w^t \in [0,1]^m$, giving reward $\ip{v^t}{w^t}$ to the algorithm. The goal of the algorithm is to maximize the sum of the rewards obtained. As usual, the reward of the algorithm is measured against the optimal fixed solution $v^*$ in hindsight. This is a generalization of classical Prediction with Experts Problem~\cite{cesaBianchi}, which corresponds to the case $q = 1$.
	
	A general connection between guarantees in the random-order (or the weaker i.i.d) model and OLO games has been recently shown in~\cite{guptaMolinaro,agrawalDevanur}. However, a crucial point is that since we \textbf{simultaneously} want worst-case guarantee as well, it is not clear that we we can employ an OLO algorithm in a \textbf{black-box fashion}. Interestingly, our new OLO algorithm has better regret than what is available in the literature, which is needed for the optimal random-order guarantee of \ultimate. We are interested in OLO algorithms with multiplicative/additive regret of the form $\algo \ge (1-\e) \OPTfixed - R$, where $\OPTfixed$ denotes the reward of the best fixed solution in hindsight. To the best of our knowledge, the best such bound for this OLO game is $\algo \ge (1-\e) \OPTfixed - O(\frac{m^{1/p} \log m}{\e})$, obtained in the seminal paper of Kalai and Vempala~\cite{kalai}. Our OLO algorithm has regret $\algo \ge(1-\e) \OPTfixed - O(\frac{p (m^{1/p} - 1)}{\e})$, see Theorem \ref{thm:OLO}; this gives a $\log m$ factor reduction in the additive term for small $p$, and dominates the Kalai-Vempala bound for all $p$.

		Another interesting connection is that we use our OLO algorithm to prove a purely probabilistic inequality (Lemma \ref{lemma:breakCorr}) that controls the correlations arising in the random-order model, a common source of difficulty for the analysis in this model. In~\cite{guptaMolinaro}, such control was obtained via a maximal inequality and union bound for the special case of the $\ell_\infty$-norm. However, for general $\ell_p$-norms a straightforward union bound gives a weaker bound than the OLO-based approach, leading to suboptimal guarantees.
		\mmnote{Put in the section: see also~\cite{rakhlin} for a relationship between OLO algorithms and concentration of measure}

	
	An important component used in both \ultimate and our OLO algorithm is a \textbf{smoothened} version $\psi_{p, \e}$ of the $\ell_p$-norm; in particular, this is what allows us to see \ultimate simultaneously as both an approximately greedy algorithm and an approximately primal-dual algorithm, as mentioned before.	 This smoothened function can be seen as a generalization of the exp-sum function $\textrm{ExpSum}(u) = \frac{1}{\e} \ln \sum_i e^{\e u_i}$, a much used smoothing of the $\ell_\infty$-norm. Given the host of applications of exp-sum, we hope that the smoothings $\psi_{p,\e}$ will find use in other contexts.

	\subsection{Roadmap and notation} \label{sec:road}
	
	In Section \ref{sec:OLO} we present and analyze our OLO algorithm \smoothbaseline, and also define the smoothing $\psi_{\e,p}$ used throughout. In the following section, we use this OLO result to prove the correlation inequality that is needed for the random-order analyses of all the algorithms considered. In Section \ref{sec:greedy} we analyze the greedy algorithm with restart \greedywr. In Section \ref{sec:improved} we present algorithm \smoothgreed, which has improved random-order guarantee but has a spurious term in its worst-case guarantee. Finally, in Section \ref{sec:ultimate} we combine this algorithm with the greedy one to remove this spurious term, obtaining algorithm \ultimate.
	
	We now define some notation.  We use $\ell_q^+$ to denote the set of non-negative vectors $\R^m_+$ with $\ell_q$ norm at most 1. Given $p \in (1, \infty)$, its \emph{\holder conjugate} $q$ is the number that satisfies $\frac{1}{p} + \frac{1}{q} = 1$. It follows from norm duality that if $p$ and $q$ are \holder conjugate, then for every vector $x \ge 0$
	\begin{align}
		\forall y \in \ell^+_q~~~~ \ip{x}{y} \le \|x\|_p, ~~~~~~\textrm{ and }~~~~~~~~~~ \|x\|_p = \max_{y \in \ell^+_q} \ip{x}{y}. \label{eq:dualNorm}
	\end{align}
 Also, we will use the well-known comparison between norms: if $p \ge p'$, then for every vector $x \in \R^m$ we have $\|x\|_p \le \|x\|_{p'} \le m^{\frac{1}{p'} - \frac{1}{p}} \|x\|_{p'}$. Finally, we use bold letters for random variables.


\section{The $\ell_q^+$ OLO problem and the Smoothened Baseline Gradient algorithm} \label{sec:OLO}

	\paragraph{$\ell_q^+$ OLO problem.}	Recall that the $\ell_q^+$ OLO problem proceeds in $n$ rounds. In round $t$, first the algorithm chooses a vector $v^t \in \ell_q^+$ based on the adversary's previous vectors $w^1, \ldots, w^{t-1}$. Then the adversary chooses a vector $w^t \in [0,1]^m$, and the algorithm obtains reward $\ip{w^t}{v^t}$. The goal of the algorithm is to maximize the sum of the rewards $\sum_{t = 1}^n \ip{w^t}{v^t}$. The \emph{regret} of the algorithm is obtained by comparing against the best fixed decision $v \in \ell_q^+$ in hindsight. We say that an algorithm has \emph{$(\e, R)$-regret} if
	\begin{align}
		\sum_{t = 1}^n \ip{w^t}{v^t} \ge e^{-\e} \left( \max_{v \in \ell_q^+} \sum_{t = 1}^n \ip{w^t}{v} - R \right). \label{eq:regret}
	\end{align}
	Recall that $e^{\pm \e}$ is approximately $(1 \pm \e)$ for small values of $\e$\short{.}{, so this is a reparametrization of the standard multiplicative/additive regret used, for instance, in~\cite{kalai}.}

	\paragraph{Smoothened Baseline Gradient Algorithm.}	To obtain an intuition about algorithms for this problem, we can see the right-hand side of the regret expression in a different way. Let $p$ be the \holder conjugate of $q$. Then duality of norms (equation \eqref{eq:dualNorm}) gives that $\max_{v \in \ell_q^+} \sum_t \ip{w^t}{v} = \|\sum_t w^t\|_p$, hence the regret expression becomes $\sum_t \ip{w^t}{v^t} \ge e^{-\e}( \|\sum_t w^t\|_p - R)$. Thus, we can interpret the algorithm's decision $v^t$ as trying to locally approximate the \emph{baseline potential} $\|.\|_p$ at $\sum_{t' < t} w^{t'}$ to capture the increase in norm caused by the unknown $w^t$. 
	
	Thus, a natural strategy is to choose $v^t$ as a (sub-)gradient $\nabla \|\sum_{t'< t} w^{t'}\|_p$ belonging to $\ell_q^+$. However, one can show that this strategy has too high regret. 
	The issue is that the gradient can quickly vary from one point to another, so approximating the value $\|u+v\|_p = \|u\|_p + \int_0^1 \nabla\|u + x v\|_p~dx$ by the first order expression $\|u\|_p + \ip{\nabla\|u\|_p}{v}$ is not good enough.	To avoid this problem, we will replace the norm $\|.\|_p$ by a smoother function $\psi$ satisfying the following:
	\begin{align}
		&\textrm{(a) (additive error) For all $u \in \R^m_+$, $\|u\|_p \le \psi(u) \le \|u\|_p + R$}& \label{eq:propSmooth1}\\
		&\textrm{(b) (stability) For all $u \in \R^m_+$ and $v \in [0,1]^m$, $\nabla \psi(u + v) \stackrel{\textrm{pointwise}}{=} e^{\pm \e} \cdot \nabla \psi(u)$} .\label{eq:propSmooth2}
	\end{align}

	To obtain such smoothing, we notice that $\|.\|_p$ is a \emph{generalized $f$-mean}, namely $\|w\|_p = f^{-1}(\sum_i f(w_i))$ for the function $f(x) = x^p$. We then define the smoothened function 
	\begin{align}
		\psi_{\e,p}(u) = f_{\e,p}^{-1}\bigg(\sum_i f_{\e,p}(u_i)\bigg),~~~~~~~~~\textrm{where } f_{\e,p}(x) = \bigg(1 + \frac{\e x}{p}\bigg)^p, \label{eq:smooth}
	\end{align}
	which written explicitly is $\psi_{\e,p}(u) = \frac{p}{\epsilon} \|\mathbf{1} + \frac{\epsilon u}{p}\|_p - \frac{p}{\epsilon}$. Notice that as $p$ goes to infinity, $f_{\e,p}(x)$ converges to $e^{\e x}$, and so $\psi_{\e,p}$ converges to the exp-sum function $\textrm{ExpSum}(w) = \frac{1}{\e} \ln \sum_i e^{\e w_i}$, a commonly used smoothing of $\ell_\infty$.
	
	One of the main properties that motivate our definition of $f_{\e,p}(x) = (1+ \frac{\e x}{p})^p$ is that its  derivative is much more stable than that of $x^p$ for $\pm 1$ perturbations: for example, $(x^p)'(0) = 0$ but $(x^p)'(1) = 1$, while $f'_{\e,p}(0) = \e$ and $f'_{\e,p}(1) = \e (1+\frac{\e}{p})^{p-1} \le \e e^\e \lesssim \e (1+\e)$). Such functions are also used for obtaining sharp estimates of moments of sums of random variables (see Section 1.5 of \cite{deLaPena}).
	
	Once we have the ``right'' definition of the smoothened function $\psi_{\e,p}$, it is not hard to prove that it satisfies properties \eqref{eq:propSmooth1}-\eqref{eq:propSmooth2}; the proof is found in Appendix \ref{app:smoothened}.

	\begin{lemma} \label{lemma:smoothened}
		Function $\psi_{\e,p}$ satisfies properties \eqref{eq:propSmooth1}-\eqref{eq:propSmooth2} with $R = \frac{p (m^{1/p}-1)}{\e}$.
	\end{lemma}

	Now we formally state the $\psi$-based \smoothbaseline algorithm for the $\ell_q^+$ OLO problem.
	
		\begin{algorithm}[H]
		\caption{\smoothbaseline}
		\begin{algorithmic}[0]         
				\State Let $p$ be such that $\frac{1}{p} + \frac{1}{q} = 1$, define $w^0 = 0$.
        \For{each time $t$}
           \State Play vector $v^t \triangleq \nabla \psi_{\e,p}(w^1 + \ldots + w^{t-1})$
           \State Observe adversarial vector $w^t \in [0,1]^p$
        \EndFor
    \end{algorithmic}
    \end{algorithm}
	
	We show that this algorithm indeed outputs a solution to the $\ell_q^+$ OLO problem (i.e. $v^t \in \ell_q^+$) with low regret.
	
	\begin{theorem} \label{thm:OLO}
		For every $p \in (1, \infty)$ the \smoothbaseline algorithm outputs a solution to the $\ell_q^+$ OLO problem with $(\e, \frac{p(m^{1/p}-1)}{\e})$-regret.
	\end{theorem}
	
	\begin{proof}
		The fact that the actions $v^t$ played belong to $\ell_q^+$ follows directly from the expression of the gradient $\nabla \psi_{\e, p}$ (see equation \eqref{eq:derPsi} in the appendix, and notice $q = \frac{p}{p-1}$). So we just bound the regret of the algorithm; to simplify the notation we drop the subscripts from $\psi_{\e,p}$ and use $s^t = w^1 + \ldots + w^t$. 
		
		We need to show $e^\e \cdot \sum_t \ip{w^t}{v^t} \ge  \|s^n\| - \frac{p(m^{1/p}-1)}{\e}$. The main idea is to relate the value obtained by the algorithm to the smoothened function $\psi$, showing
		\begin{align}
			\textstyle e^\e \cdot \sum_{t=1}^n \ip{w^t}{v^t} \ge \psi(s^n) - \psi(0). \label{eq:regretOLO}
		\end{align}
		 First, the convexity of $\|.\|_p$ directly implies that $\psi$ is convex. Thus, for every time step $t$ we have $\psi(s^{t-1}) \ge \psi(s^t) + \ip{\nabla \psi(s^t)}{-w^t}$, or equivalently $\psi(s^t) - \psi(s^{t-1}) \le \ip{\nabla \psi(s^t)}{w^t}$. Since Lemma \ref{lemma:smoothened} guarantees that $\psi$ satisfies the gradient stability property \eqref{eq:propSmooth2}, we can upper bound the right-hand side of this expression to obtain $\psi(s^t) - \psi(s^{t-1}) \le e^\e \ip{\nabla \psi(s^{t-1})}{w^t} = e^\e \ip{v^t}{w^t}$. Adding over all $t$'s then gives inequality \eqref{eq:regretOLO}.
		 
		 From Lemma \ref{lemma:smoothened} we have the comparison $\psi(s^n) \ge \|s^n\|_p$, and notice that $\psi(0) = \frac{p (m^{1/p} - 1)}{\e}$;
		 employing these observation to inequality \eqref{eq:regretOLO} gives $e^\e \cdot \sum_t \ip{w^t}{v^t} \ge \|s^n\|_p - \frac{p(m^{1/p}-1)}{\e}$, thus concluding the proof.
	\end{proof}

	We remark that the idea of using the gradient of a smoothened baseline to obtain a low regret OLO algorithm was already used in~\cite{abernethy14}. However, our notion of smoothness is different from the ones they used (partially because we are interested in multiplicative/additive regret), and their results cannot be directly applied to obtain the regret of Theorem \ref{eq:regretOLO}.


	\section{Handling correlations of the random-order model} 
	
Informally, one of the difficulties of analyzing algorithms in the random-order model is that, unlike in the i.i.d. model, there are correlations between jobs in different time steps because they are being sampled \textbf{without replacement} from the underlying collection of jobs. In this section we control the correlations of vectors in the random-order model, which will be crucial for analyzing algorithms for the \GLB problem. Interestingly, we use the OLO algorithm \smoothbaseline to prove this purely probabilistic inequality (see~\cite{rakhlin} for another connection between OLO algorithms and martingale concentration inequalities).
	
	\begin{lemma} \label{lemma:breakCorr}
		Consider a set of vectors $\{y^1, \ldots, y^n\} \in [0,1]^m$ and let $\bY^1, \ldots, \bY^t$ be sampled without replacement from this set. Let $\bZ$ be a random vector in $\ell_q^+$ that depends only on $\bY^1, \ldots, \bY^{t-1}$. Then for all $\e > 0$, $$\E \ip{\bY^t}{\bZ} \le e^\e \|\E \bY^t\|_p + \frac{1}{n - (t-1)} \cdot \frac{p (m^{1/p} -1)}{\e}.$$
	\end{lemma}
	
	(Recall $\E\bY^t$ denotes the vector obtained by taking component-wise expectation.)
	To understand the meaning of this lemma, notice that if $z$ is a \textbf{fixed} vector in $\ell^+_q$ (or simply independent of $\bY^t$), then $\E \ip{\bY^t}{z} = \ip{\E \bY^t}{z} \le \|\E \bY^t\|_p$. On the other hand, if $\bZ$ is highly correlated to $\bY^t$, say $\bZ = \frac{\bY^t}{\|\bY^t\|_q}$, then we only have $\E \ip{\bY^t}{\bZ} = \E \|\bY^t\|_p$, which in general can be arbitrarily larger than $\|\E \bY^t\|_p$ (e.g. if $\E \bY^t = 0$). 

	The main element for proving Lemma \ref{lemma:breakCorr} is to show that because $\bY^t$ is \textbf{bounded and non-negative}, actually $\E \|\bY^t\|_p \approx \|\E \bY^t\|_p$; more precisely, we show $\E \|\bY^1 + \ldots + \bY^t\|_p \lesssim \|\E \bY^1 + \ldots + \E \bY^t\|_p$. This was proved in \cite{guptaMolinaro} for the special case $p = \infty$ using a maximal inequality, but can also be proved using Bernstein's inequality to obtain concentration for each coordinate of the sum $\bY^1+ \ldots + \bY^t$, taking a union bound to obtain concentration of the norm $\|\bY^1+ \ldots + \bY^t\|_p$, and then integrating its tail. However, for general $p$ the union bound is loose and bound obtained has an extra $\log m$ factor. We use the OLO algorithm \smoothbaseline and Hoeffding's Comparison Lemma\mmnote{Add this to intro?} to quickly provide a bound without such extra factor.
	
	\begin{lemma} \label{lemma:corrSum}
		Consider a set of vectors $\{y^1, \ldots, y^n\} \in [0,1]^m$ and let $\bY^1, \ldots, \bY^{\kappa}$ be sampled without replacement from this set. Then for all $\e > 0$ $$\textstyle\E \|\bY^1 + \ldots + \bY^\kappa\|_p \le e^{\e}\, \|\E \bY^1 + \ldots + \E \bY^\kappa\|_p + \frac{p (m^{1/p} - 1)}{\e}$$
	\end{lemma}
	
	\begin{proof}
		We show it suffices to prove the result for i.i.d vectors. Let $\tilde{\bY}^1, \ldots, \tilde{\bY}^{\kappa}$ be i.i.d sampled \textbf{with replacement} from $\{y^1, \ldots, y^n\}$. Then $\tilde{\bY}^t$ has the same expectation as $\bY^t$ and hence $\|\E \bY^1 + \ldots + \E \bY^{\kappa}\|_p = \|\E \tilde{\bY}^1 + \ldots + \E \tilde{\bY}^{\kappa}\|_p$. Moreover, Hoeffding's Comparison Lemma \cite{hoeffding,gross} gives that for every continuous convex function $f$, $\E f(\bY^1 + \ldots + \bY^{\kappa}) \le \E f(\tilde{\bY}^1 + \ldots + \tilde{\bY}^{\kappa})$; since the norm $\|.\|_p$ satisfies these properties, $\E\|\bY^1 + \ldots + \bY^{\kappa}\|_p \le \E\|\tilde{\bY}^1 + \ldots + \tilde{\bY}^{\kappa}\|_p$. Thus, it suffices to prove the lemma for the i.i.d variables\short{.}{ $\tilde{\bY}^1, \ldots, \tilde{\bY}^{\kappa}$.}
		
		For that, run the OLO algorithm \smoothbaseline over the input sequence $\tilde{\bY}^1, \ldots, \tilde{\bY}^{\kappa}$, letting $\bZ^1, \ldots, \bZ^{\kappa}$ be the vectors played by the algorithm. Using the guarantee of this algorithm (Theorem \ref{thm:OLO}) for every scenario and taking expectations, we have
		\begin{align}
			\sum_{t \le {\kappa}} \E \ip{\tilde{\bY}^t}{\bZ^t} \ge e^{-\e} \left(\E \bigg\|\sum_{t \le {\kappa}} \tilde{\bY}^t\bigg\|_p - \frac{p (m^{1/p} - 1)}{\e}\right) \label{eq:corr1}
		\end{align}
		Since $\bZ^t$ only depends on $\tilde{\bY}^1, \ldots, \tilde{\bY}^{t-1}$ and $\tilde{\bY}^t$ is independent from these variables, we have $\E \ip{\tilde{\bY}^t}{\bZ^t} = \ip{\E\tilde{\bY}^t}{\E\bZ^t}$. Moreover, the $\tilde{\bY}^t$'s are identical, so their expectations equal $\frac{1}{n}$ of $\mu := \E \tilde{\bY^1} + \ldots + \E \tilde{\bY^{\kappa}}$, \short{hence}{and hence}
		\begin{align*}
			\textstyle\sum_{t \le \kappa} \E \ip{\tilde{\bY}^t}{\bZ^t} = \sum_{t \le \kappa} \ip{\frac{\mu}{n}}{\E \bZ^t} =   \ip{\mu}{\frac{1}{n} \sum_{t\le \kappa}\E \bZ^t} \le \|\mu\|_p,
		\end{align*}
		where the last inequality follows from the fact that $\|\frac{1}{n} \sum_{t \le \kappa} \bZ^t\|_q \le 1$ (since each $\|\bZ^t\|_q \le 1$ and $\|.\|_q$ is convex) and inequality \eqref{eq:dualNorm}. Employing this \short{}{bound} on inequality \eqref{eq:corr1} gives the desired inequality and concludes the proof.
	\end{proof}

	\begin{proof}[of Lemma \ref{lemma:breakCorr}]
		Let $\E_{t-1}$ denote the expectation conditioned on $\bY^1, \ldots, \bY^{t-1}$. We break the expectation as $\E \ip{\bY^t}{\bZ} = \E \E_{t-1} \ip{\bY^t}{\bZ}$. Again, since $\bZ$ is determined by $\bY^1, \ldots, \bY^{t-1}$ and belongs to $\ell^+_q$, we get $\E_{t-1} \ip{\bY^t}{\bZ} = \ip{\E_{t-1} \bY^t}{\bZ} \le \|\E_{t-1} \bY^t\|_p$; thus it suffices to upper bound this quantity in expectation.
		
		Let $\mu = \frac{1}{n} \sum_{t'} y^{t'}$ be the average of the vectors. Since $\bY^t$ is uniformly sampled from the vectors that have not appeared in the samples $\bY^1, \ldots, \bY^{t-1}$, we have that its conditional expectation is $$\E_{t-1} \bY^t = \frac{n \mu - (\bY^1 + \ldots + \bY^{t-1})}{n - (t-1)}.$$ Moreover, notice that $n \mu - (\bY^1 + \ldots + \bY^{t-1})$ (i.e. the sum of the remaining $n - (t-1)$ vectors) has the same distribution as $\bY^1 + \ldots + \bY^{n - (t-1)}$, so $\E_{t-1} \bY^t$ has the same distribution as $\frac{\bY^1 + \ldots + \bY^{n - (t-1)}}{n - (t-1)}$. Then using Lemma \ref{lemma:corrSum} we can upper bound the expected value of $\|\E_{t-1} \bY^t\|$ as
		\begin{align*}
			\E \|\E_{t-1} \bY^t\|_p \le \frac{1}{n-(t-1)} \left[e^{\e}\,\|(n - (t-1) \mu\|_p + \frac{p (m^{1/p}- 1)}{\e} \right]
		\end{align*}  
		for all $\e \in (0,1]$. Reorganizing this expression concludes the proof.
	\end{proof}


	\section{Greedy algorithm for the \GLB problem} \label{sec:greedy}

	Now we return to our main problem of interest, the \lpGLoadBal problem, defined in the introduction. In this section we consider the greedy algorithm with restart at time $n/2$, which can be more formally described as follows:
	\begin{algorithm}[H]
		\caption{\greedywr}
		\begin{algorithmic}[0]
				\For{time $t = 1, \ldots, \frac{n}{2}$}
					\State Select $\bar{x}^t$ to minimize the load $\|\sum_{\tau = 1}^t A^\tau \bar{x}^\tau\|_p$
				\EndFor
				\For{time $t = \frac{n}{2} + 1, \ldots, n$}
					\State Select $\bar{x}^t$ to minimize the load $\|\sum_{\tau = \frac{n}{2} + 1}^t A^\tau \bar{x}^\tau\|_p$
				\EndFor
    \end{algorithmic}
    \end{algorithm}
    
		Also recall that Theorem \ref{thm:greedy} presented in the introduction states the worst-case and random-order guarantees of this algorithm; in the remainder of this section we prove this theorem. 
		
		Since the greedy algorithm without restart is $O(p)$-competitive in the worst-case, it is straightforward to show that \greedywr also inherits this guarantee: by triangle inequality, the load of the algorithm is $\|\sum_t A^t \bar{x}^t\|_p \le \|\sum_{t \le n/2} A^t \bar{x}^t\|_p + \|\sum_{t > n/e} A^t \bar{x}^t\|_p$; but these terms are respectively at most $O(p)$ times the optimal load for the first and second half of the instance, each of which is at most $\OPT$, thus concluding the argument. Therefore, it suffices to analyze the random-order behavior of the algorithm, proving part (b) of the theorem.
    
	So we use $\bA^t$ to denote the random matrix that arrives at time $t$ and $\bar{\bx}^t$ to denote the random fractional assignment output by \greedywr. Also let $\tilde{\bx}^t$ be the optimal offline decision for time $t$.\footnote{More formally, let $\{\tilde{x}^t\}_t$ be an optimal solution for the offline instance, and let $\bm{\sigma}$ be the random permutation of $[n]$ such that $\bA^t = A^{\bm{\sigma}(t)}$; then define $\tilde{\bx}^t := \tilde{x}^{\bm{\sigma}(t)}$.}
	  Because of the restart, and random order, the load vectors obtained by \greedywr in the first and second half of the process, namely $\sum_{t \le n/2} \bA^t \bar{\bx}^t$ and $\sum_{t > n/2} \bA^t \bar{\bx}^t$, have the same distribution. 
	Again due to triangle inequality, it then suffices to analyze the first half and show that
	\begin{align}
		\E\,\bigg\|\sum_{t \le n/2} \bA^t \bar{\bx}^t\bigg\|_p \le \bigg(\frac{e^{\e}}{2} + \e\bigg) \OPT + m^{1/p} + \frac{(p-1)\, m^{1-1/p}}{2\e} + \frac{p\, m^{1/p}}{\e}; \label{eq:greedyHalf}
	\end{align}
	notice that this implies the bound in Theorem \ref{thm:greedy} because $e^\e \le 1+2\e$ for $\e \in [0,1]$ and $m^{1/p} \le m^{1-1/p}$ for $p  \ge 2$.
	To simplify the notation, let $\bS^t = \sum_{t' \le t} \bA^{t'} \bar{\bx}^{t'}$ be the random load vector of \greedywr up to time $t$.
	
	The main tool for analyzing the load increments $\|\bS^t\|_p - \|\bS^{t-1}\|_p$ incurred by the algorithm is the following estimate for the $\ell_p$-norm. One of its crucial features is that it shows that the linearization of the $\ell_p$-norm is increasingly better as we move away from the origin. It is a quick corollary of the optimal modulus of strong smoothness of the square of the $\ell_p$-norm recently proved in~\cite{AMS}, and is proved in Appendix \ref{app:proofLinLp}.
   
   \begin{lemma}\label{lemma:linLp}
   		Consider $p \in [2, \infty)$ and let $q$ be its \holder conjugate. Then for every non-negative vectors $u \in \R^n_+ \setminus \{0\}$ and $v \in \R^n_+$, there is a vector $g(u) \in \R^n_+$ with $\|g(u)\|_q \le 1$ such that $$\textstyle\|u + v\|_p \le \|u\|_p + \ip{g(u)}{v} + \frac{(p-1) \|v\|^2_p}{2 \|u\|_p}.$$
   \end{lemma}

   Now we analyze algorithm \greedywr. We handle separately the initial time steps where the load is small, so define the stopping time $\btau = \min\{ t \le n/2 : \|\bS^t\|_p > \e \OPT \}$ (set $\btau = n/2$ for the scenarios with $\|\bS^{n/2}\|_p \le \e \OPT$), load of the algorithm up to time $n/2$ can be written as $\|\bS^{n/2}\|_p = \|\bS^{\btau}\|_p + \sum_{t = \btau + 1}^{n/2} (\|\bS^t\|_p - \|\bS^{t-1}\|_p)$. From the greedy property we have the load $\|\bS^t\|_p$ is at most $\|\bS^{t-1} + \tilde{\bell}^t\|_p$, where $\tilde{\bell}^t = \bA^t \tilde{\bx}^t$ is the load incurred by $\OPT$ at time $t$. Thus, employing the estimate from Lemma \ref{lemma:linLp} we get
		\begin{align}
			\|\bS^{n/2}\|_p &\le \|\bS^{\btau}\|_p + \sum_{t = \btau + 1}^{n/2} \ip{g(\bS^{t-1})}{\tilde{\bell}^t} + \sum_{t = \btau + 1}^{n/2} \frac{(p-1) \|\tilde{\bell}^t\|^2_p}{2 \|\bS^{t-1}\|_p}. \label{eq:greedyMain}
		\end{align}
		We upper bound each term of the right-hand side separately.
		
		\paragraph{First term of RHS of \eqref{eq:greedyMain}.} Since $\tau$ is the first time $\|\bS^{\btau}\|_p$ goes above $\e \OPT$, and the load does not increase by more than $m^{1/p}$ per time step (which uses the fact that the entries of the matrices $\bA^t$ are in $[0,1]$), we have $\|\bS^{\btau}\|_p \le \|\bS^{\btau - 1}\|_p + \|\bA^{\btau} \bar{\bx}^{\btau}\|_p \le \e \OPT + m^{1/p}$. 
		
		\paragraph{Last term of RHS of \eqref{eq:greedyMain}.} First notice that since we are only adding terms after the stopping time $\btau$, each denominator will be at least $2 \e \OPT$; so we have the upper bound $\frac{(p-1)}{2\e\OPT} \sum_{t = \btau +1}^{n/2} \|\tilde{\bell}^t\|^2_p$. To bound this remaining sum, we will linearize it by passing to the $\ell_1$ norm and them back to $\ell_p$: Since all entries of the load vector $\tilde{\bell}$ are in $[0,1]$, we have that $\|\tilde{\bell}^t\|_p^2 \le \|\tilde{\bell}^t\|_1$:
		\begin{align*}
			\|\tilde{\bell}^t\|^2_p = \bigg(\sum_i \big(\tilde{\bell}^t_i\big)^p \bigg)^{2/p} \le  \bigg(\sum_i \big(\tilde{\bell}^t_i\big)^{p/2} \bigg)^{2/p} = \|\tilde{\bell}^t\|_{p/2} \le \|\tilde{\bell}^t\|_1.
		\end{align*}
		\short{}{\red{Missing comma at the end of eq.}where the last inequality uses standard comparison of norms.}
		Moreover, the non-negativity of these vectors give additivity for $\|.\|_1$, namely $\sum_{t \le n/2} \|\tilde{\bell}^t\|_1 = \|\sum_{t \le n/2} \tilde{\bell}^t\|_1$. Finally, by \short{}{standard} comparison of norms this is at most $m^{1-1/p} \|\sum_{t \le n/2} \tilde{\bell}^t\|_p$. Thus, the last term of \eqref{eq:greedyMain} can be upper bounded by $\frac{(p-1) m^{1-1/p}}{2 \e \OPT} \|\sum_{t \le n/2} \tilde{\bell}^t\|_p \le \frac{(p-1) m^{1-1/p}}{2 \e}$.
		
		\paragraph{Second term of RHS of \eqref{eq:greedyMain}.} This is the main term in the RHS of \eqref{eq:greedyMain}, and we need to show that \emph{in expectation} it is about at most $\OPT$; this is the only place we use the random-order model and that the algorithm restarts at $n/2$. Since $g(\bS^{t-1})$ only depends on items seen up to time $t-1$, we can employ Lemma \ref{lemma:breakCorr} to obtain that $\E\ip{g(\bS^{t-1})}{\tilde{\bell}^t} \le e^\e  \|\E \tilde{\bell}^t\|_p + \frac{1}{\e} \cdot \frac{p (m^{1/p} -1)}{n - (t-1)}$. Moreover, notice that \OPT's expected load $\E\tilde{\bell}^t$ is the same in every time step, and so $\|\E \tilde{\bell}^t\|_p = \frac{1}{n} \|\E \sum_{t = 1}^n \tilde{\bell}^t\|_p = \frac{1}{n} \OPT$.\mmnote{Add more explanation and move this somewhere else? Probably needed in the other section as well} Since we are only considering $t \le n/2$, our expression can be further bounded as $\E\ip{g(\bS^{t-1})}{\tilde{\bell}^t} \le  \frac{e^\e}{n} \OPT + \frac{2 p (m^{1/p} -1)}{\e n}$. Adding over all these time steps we get $\E \sum_{t = \btau + 1}^{n/2} \ip{g(\bS^{t-1})}{\tilde{\bell}} \le \frac{e^\e}{2} \OPT + \frac{p (m^{1/p} - 1)}{\e}$.
		
		\medskip Employing all these bound in inequality \eqref{eq:greedyMain} proves \eqref{eq:greedyHalf}. This concludes the proof of Theorem \ref{thm:greedy}.


	\section{Towards improved simultaneous guarantees: algorithm \smoothgreed} \label{sec:improved}
	
	We now present the algorithm \smoothgreed that has improved random-order guarantee at the expense of a slightly suboptimal worst-case guarantee.
	

		\begin{algorithm}[H]
		\caption{\smoothgreed$(p, \e)$}
		\begin{algorithmic}[0]
				\State Let $\psi_{\e,p}(u) = \frac{p}{\epsilon} \|\mathbf{1} + \frac{\epsilon u}{p}\|_p - \frac{p}{\epsilon}$ (as in equation \eqref{eq:smooth}).
				\For{time $t = 1, \ldots, \frac{n}{2}$}
					\State Select $\bar{x}^t \in \Delta^m$ to minimize $\psi_{\e,p}(\sum_{\tau = 1}^t A^\tau \bar{x}^\tau)$
				\EndFor
				\For{time $t = \frac{n}{2}+1, \ldots, n$}
					\State Select $\bar{x}^t \in \Delta^m$ to minimize $\psi_{\e,p}(\sum_{\tau = \frac{n}{2}+1}^t A^\tau \bar{x}^\tau)$
				\EndFor
    \end{algorithmic}
    \end{algorithm}

	The motivation behind this algorithm is the following: First, since this is simply \greedywr on the modified function $\psi_{\e,p}(.) = \frac{p}{\epsilon} \|\mathbf{1} + \frac{\epsilon u}{p}\|_p - \frac{p}{\epsilon}$, it is intuitive that it approximately inherits the worst-case guarantee of \greedywr. On the other hand, the smoothness of $\psi_{\e,p}$ (equation \eqref{eq:propSmooth2}), guarantees that its gradient captures well its behavior, so \smoothgreed is almost greedy on this gradient; this allow us to connect the algorithm with the \smoothbaseline OLO algorithm to provide guarantees in the random-order model. Here is the formal guarantees of \smoothgreed.
	
	\begin{theorem} \label{thm:smoothGreedy}
		For all $p \in [2, \infty)$ algorithm \smoothgreed has the following guarantees for any $\e > 0$
		\vspace{-5pt}
		\begin{enumerate}
			\item[(a)] In the worst-case model it has load at most $O(p) \OPT + \frac{4 p (m^{1/p} - 1)}{\e}$ 
			\vspace{-10pt}
			\item[(b)] In the random-order model has expected load at most $e^{2\e} (\OPT + \frac{4 p (m^{1/p} - 1)}{\e})$.
		\end{enumerate}
	\end{theorem}	


	\short{\paragraph{Analysis in the worst-case model.}}{\subsection{Analysis in the worst-case model}}
	We prove part (a) of Theorem \ref{thm:smoothGreedy}, where the idea is connect with \greedywr by adding extra jobs to the input. 	
	
	Assume the number of machines $m \ge 2$, otherwise 	 it is easy to see that \smoothgreed is optimal.	Consider an arbitrary input sequence $A^1, \ldots, A^n$ for \smoothgreed. Define $B^1 =  \ldots =  B^w$ to be the all 1's $m \times k$ matrix, for $w = \frac{p}{\e}$ (we assume for simplicity that $w$ is integral). Since $\psi_{\e,p}(u) = \|\frac{p}{\e} \ones + u\|_p - \frac{p}{\e} = \|(B^1 y^1 + \ldots + B^w y^w) + u\|_p - \frac{p}{\e}$ for any $y^t$'s, it is easy to see that behavior of \smoothgreed up to time $n/2$ is the same as that of \greedy (without restart, over norm $\|.\|_p$) on input $B^1, \ldots, B^w, A^1, \ldots, A^{\frac{n}{2}}$, and the same for time periods $n/2, \ldots, n$. Now one can directly employ the standard $O(p)$-guarantee for \greedy to get that \smoothgreed has load at most $O(p) (\OPT + \frac{p (m^{1/p} - 1)}{\e})$; however this leads to an additive error that is \textbf{quadratic} in $p$. To obtain an improved bound, we use the following more refined guarantee for \greedy, which can be obtained from the analysis in \cite{caragiannis} (we present a proof in Appendix \ref{app:refinedGuarantee}).
	
	\begin{lemma} \label{lemma:refinedGuarantee}
		Consider an arbitrary sequence of jobs $C^1, \ldots, C^n$, let $\{\bar{x}^t\}_t$ be the actions output by \greedy over $\|.\|_p$ and let $\{\tilde{x}^t\}_t$ be the optimal solution. Then for all $\tau$
		\begin{align*}
			\bigg\|\sum_{t} C^{t} \bar{x}^{t}\bigg\|_p - 2^{1/p} \bigg\|\sum_{t = 1}^{\tau-1} C^{t} \bar{x}^{t}\bigg\|_p \le O(p) \cdot \bigg\|\sum_{t = \tau}^n C^{t} \tilde{x}^{t}\bigg\|_p.
		\end{align*}
	\end{lemma}
	
	Then let $S^t$ be the total load vector obtained by \smoothgreed up to time $t$. Using triangle inequality and then Lemma \ref{lemma:smoothened} we decompose the load of the algorithm \short{$\|S^n\|_p \le \|S^{\frac{n}{2}}\|_p + \|S^n - S^{\frac{n}{2}}\|_p \le \psi_{\e,p}(S^{\frac{n}{2}}) + \psi_{\e,p}(S^n - S^{\frac{n}{2}})$.}{
	have that the load of the algorithm can be decomposed
	\begin{align}
		\|S^n\|_p &\le \|S^{\frac{n}{2}}\|_p + \|S^n - S^{\frac{n}{2}}\|_p \le \psi_{\e,p}(S^{\frac{n}{2}}) + \psi_{\e,p}(S^n - S^{\frac{n}{2}}). \label{eq:wcSmoothGreedy}
	\end{align}
	}
	To upper bound the term $\psi_{\e,p}(S^{\frac{n}{2}})$, we apply Lemma \ref{lemma:refinedGuarantee} with $\{B^t\}_{t}$ corresponding to $\{C^{t}\}_{t < \tau}$ and $\{A^t\}_t$ corresponding to $\{C^t\}_{t \ge \tau}$ to get
	\begin{align*}
		\psi_{\e,p}(S^{\frac{n}{2}}) = \left\|\frac{p}{\e} \ones + S^{\frac{n}{2}}\right\|_p - \frac{p}{\e} \le O(p) \, \OPT_{1} + 2^{1/p} \left\|\frac{p}{\e} \ones\right\|_p - \frac{p}{\e} \le O(p) \, \OPT + \frac{p ((2 m)^{1/p} - 1)}{\e},
	\end{align*}
	where $\OPT_1$ is the optimal load up to time $n/2$. Moreover, the function $x \mapsto x^{1/p} - 1$ is subadditive over $[0, \infty)$ (since it is non-negative, concave and has value 0 at the origin~\cite[Theorem 7.2.5]{hillePhillips}), and hence the last term of the right-hand side is at most $\frac{p}{\e}((m^{1/p} - 1) + (m^{1/p} - 1)) = \frac{2p (m^{1/p} - 1)}{\e}$. Similarly, for the second half \short{}{of the instance} we have $\psi_{\e,p}(S^n - S^{\frac{n}{2}}) \le O(p)\, \OPT + \frac{2p (m^{1/p} - 1)}{\e}$. Employing these bounds \short{proves}{on inequality \eqref{eq:wcSmoothGreedy} then proves} part (a) of Theorem~\ref{thm:smoothGreedy}.
	
   

 		\short{\paragraph{Analysis in the random-order model.}}{ \subsection{Analysis in the random-order model}}
    Now we analyze algorithm \smoothgreed in the random-order model, proving part (b) of Theorem \ref{thm:smoothGreedy}. As usual, let $\bA^1, \ldots, \bA^n$ be the sequence of jobs in random order, $\{\bar{\bx}^t\}_t$ be the decisions output by \smoothgreed, and $\bS^t = \sum_{\tau\le t} \bA^{\tau} \bar{\bx}^{\tau}$ be the load vector up to time $t$. 

		\mmnote{Add this discussion (over the next 3 par) to intro}The first main idea for the analysis is that algorithm \smoothgreed is also ``approximately greedy'' with respect to the gradient $\nabla \psi_{\e,p}$. This is due to the smoothness property \eqref{eq:propSmooth2}, and a main reason for defining \smoothgreed as greedy over the smoothened function $\psi_{\e,p}$ instead of the original one $\|.\|_p$; this lemma follows directly by integrating property \eqref{eq:propSmooth2} (see Appendix \ref{app:greedyGrad}).
		
    
    \begin{lemma} \label{lemma:greedyGrad}
    	For \short{}{any }$u \in \R^m_+$ and $v,v' \in [0,1]^m$, if $\psi_{\e,p}(u + v) \le \psi_{\e, p}(u + v')$ then $\ip{\nabla \psi_{\e,p}(u)}{v} \le e^{2\e} \ip{\nabla \psi_{\e,p}(u)}{v'}$.
    \end{lemma} 
    
    Because of that, forgetting about the restart for now, algorithm \smoothgreed can be seen as an approximation to the primal-dual-type algorithms of \cite{guptaMolinaro,agrawalDevanur}: Considering the expression $\sum_t \ip{\nabla \psi_{\e,p}(\bS^{t-1})}{\bA^t \bar{\bx}^t}$, on the primal view the algorithm is choosing $\bar{\bx}^t$ to approximately minimize this expression online, and on the dual view the gradient $\nabla \psi_{\e,p}(\bS^{t-1})$ is playing the role of dual variables trying to maximize this expression. 
    
    The second crucial point of using $\psi_{\e,p}$ is that these dual variables $\nabla \psi_{\e,p}(\bS^{t-1})$ are exactly being played according to algorithm \smoothbaseline (over input $\bA^1, \ldots, \bA^n$), which we showed in Theorem \ref{thm:OLO} has small regret. This means that the sum $\sum_t \ip{\nabla \psi_{\e,p}(\bS^{t-1})}{\bA^t \bar{\bx}^t}$ is approximately capturing the actual load $\|\bS^n\|_p$. From this point on, the analysis follows the same lines as that of \cite{guptaMolinaro,agrawalDevanur}.
    
    To make this more formal, let $\bg^t$ be the gradient used by algorithm \smoothgreed at time $t$, namely 
    $\bg^t = \nabla \psi_{\e,p}(\bS^{t-1})$
     for $t \le n/2$ and 
 $\bg^t = \nabla \psi_{\e,p}(\bS^{t-1} - \bS^{\frac{n}{2}})$
 for $t > \frac{n}{2}$, and let $R := \frac{p (m^{1/p} - 1)}{\e}$\mmnote{Check constants later} be the additive regret in Theorem \ref{thm:OLO}. Using Theorem \ref{thm:OLO} in the two halves of  algorithm \smoothgreed in every scenario, we have that 
  $\sum_{t} \ip{\bA^t \bar{\bx}^t}{\bg^t} \ge e^{-\e} (\|\bS^{\frac{n}{2}}\|_p + \|\bS^n - \bS^{\frac{n}{2}}\|_p) - 2R \ge e^{-\e} \|\bS^n\|_p -2R$, where the last inequality follows from triangle inequality; thus, the load incurred by the algorithm satisfies
     \begin{align}
     	 \|\bS^n\|_p \le e^\e \left(\sum_t \ip{\bA^t \bar{\bx}^t}{\bg^t} + 2R\right). \label{eq:RO1}
     \end{align}
       Now we upper bound the right-hand side in expectation.
       Because of the restart of the algorithm, and the random order, the distribution of the first half $(\ip{\bA^t \bar{\bx}^t}{\bg^t})_{t=1}^{\frac{n}{2}}$ is the same as that of the second half $(\ip{\bA^t \bar{\bx}^t}{\bg^t})_{t=\frac{n}{2}+1}^{n}$, and so it suffices to bound the first half sum $\E \sum_{t\le\frac{n}{2}} \ip{\bA^t \bar{\bx}^t}{\bg^t}$.

 Let $\tilde{\bx}^1, \ldots, \tilde{\bx}^n$ be the optimal offline solution. By the greedy criterion of \smoothgreed and Lemma \ref{lemma:greedyGrad}, the algorithm has almost better load than the optimal solution when measured though the $\bg^t$'s: $\sum_t \ip{\bA^t \bar{\bx}^t}{\bg^t} \le e^{2\e} \sum_t \ip{\bA^t \tilde{\bx}^t}{\bg^t}$. But since the gradient $\bg^t$ is determined by $\bS^{t-1}$ and belongs to $\ell_q^+$, Lemma \ref{lemma:breakCorr} gives the upper bound $\E \ip{\bA^t \tilde{\bx}^t}{\bg^t} \le e^\e \|\E \bA^t \tilde{\bx}^t\|_p + \frac{1}{\e} \cdot \frac{p (m^{1/p} - 1)}{n - (t-1)}$; also notice that $\|\E\bA^t \tilde{\bx}^t\|_p = \frac{1}{n} \OPT$. Adding these bounds over all $t \le \frac{n}{2}$, we obtain 
	\vspace{-5pt}
	\begin{align*}
		\E \sum_{t \le \frac{n}{2}} \ip{\bA^t \bar{\bx}^t}{\bg^t} \le e^{\e} \left( \frac{\OPT}{2} + \frac{p (m^{1/p} - 1)}{\e} \right).
	\end{align*}
  Plugging this bound in inequality \eqref{eq:RO1} we get that the expected load of the algorithm is $\E\|\bS^n\|_p \le e^{2\e} \left(\OPT + 4R\right)$. This concludes the proof of part (b) of Theorem \ref{thm:smoothGreedy}.


	\section{Algorithm \ultimate} \label{sec:ultimate}
	
	Since algorithm \smoothgreed incurs an additive error in the worst-case, if $\OPT$ is small it may not give the desired $O(p)$ multiplicative guarantee. Thus, the idea is to use the regular greedy algorithm until the accumulated load is large enough, and then switch to \smoothgreed.
	

		\begin{algorithm}[H]
		\caption{\ultimate$(p, \e)$}
		\begin{algorithmic}[0]
			\State 1. Run algorithm \greedy, obtaining solution $\bar{x}^1, \ldots, \bar{x}^{\bar{t}}$, until the first time $\bar{t}$ the load $\|\sum_{t \le \bar{t}} A^t \bar{x}^t\|_p$ obtained becomes larger than $\frac{p (m^{1/p} -1)}{\e}$.
    	\State 2. Run algorithm \smoothgreed over the remaining $n-\bar{t}$ time steps, obtaining $\bar{x}^{\bar{t} + 1}, \ldots, \bar{x}^{n}$.
    \end{algorithmic}
    \end{algorithm}

	The guarantees of the algorithm are given in Theorem \ref{thm:ultimate}, which we now prove. By triangle inequality the load of the algorithm is at most $$\textstyle \|\sum_{t < \bar{t}} A^t \bar{x}^t\|_p + \|A^t \bar{x}^t\|_p + \|\sum_{t > \bar{t}} A^t \bar{x}^t\|_p ~\le~ \frac{p(m^{1/p} -1)}{\e} + m^{1/p} + \|\sum_{t > \bar{t}} A^t \bar{x}^t\|_p \le \frac{2 p(m^{1/p} -1)}{\e} + \|\sum_{t > \bar{t}} A^t \bar{x}^t\|_p ,$$ where the last inequality uses $p \ge 2$ and $\e \le 1$.
	
	In the random-order model, notice that when we condition on the stopping time $\bar{t}$ and the items seen thus far, the items in the remaining times $\bar{t} +1, \ldots, n$ are a random permutation of the remaining items, so we can \short{}{simply} apply the guarantee of \smoothgreed to bound the last term of the displayed inequality, giving part (b) of the theorem.
		
	In the worst-case model, we can also apply the guarantee of \smoothgreed to bound this term, and further note that $\frac{m^{1/p} - 1}{\e} \le O(\OPT)$: since the algorithm incurs load more than $\frac{p (m^{1/p} - 1)}{\e}$ up to time $\bar{t}$ and \greedy is $O(p)$-approximate, the optimal load up to time $\bar{t}$ is at least $\Omega(\frac{m^{1/p} - 1}{\e})$, which lower bounds $\OPT$. This gives part (a) of the theorem, and concludes the analysis of \ultimate.


    %
%


	\bibliographystyle{alpha}
	\bibliography{online-lp-short}

	\pagebreak
	\appendix
	\noindent \textbf{\LARGE Appendix}
	\\
	

	\section{Proof of Theorem \ref{thm:LBWC}} \label{app:LBWC}

	The lower bound is essentially the same used in \cite{awerbuch}, we just need to work with the \emph{expected load} incurred.	Fix a positive integer $M$, let $m = 2^{p + 1}$ (for simplicity we assume $p$ integral\mmnote{Maybe ``due to [comparison of norm]'', if have it}) and fix an algorithm. We will construct an instance where each job can be processed in a subset of the machines and incurs load 1 in the machine chosen to process it. So given a subset $S \subseteq [m]$ of machines, a \emph{type $S$} job is one that can be processed by the machines in $S$, again incurring load 1 in the chosen machine. 
	
	 The instance is constructed in $\log m$ rounds. In round $i$, we have machines $\U_{i-1} \subseteq [m]$ ``active''. These active machines are paired up into $|\U_{i-1}|/2$ disjoint pairs, and the adversary sends $M$ copies of jobs of types corresponding to each such pair $\{a,b\}$. Let $\bell^i$ be the (randomized) load vector incurred by the algorithm when processing these round-$i$ jobs. Then for each such pair $\{a,b\}$, the machine $j \in (a,b)$ with smallest load $\E \bell^i_j$ (ties broken arbitrarily) is deactivated, defining the next active set $\U_{i}$. Notice that for all  machines $j \in \U_{i}$ we have $\E \bell^i_j \ge \frac{M}{2}$. This proceeds until round $\log m$.

	
	We analyze this instance starting with $\OPT$. Consider the following strategy: process all round-$i$ by spreading them uniformly over the machines $\U_{i-1} \setminus \U_i$, incurring load $M$ in each of them. By construction, the machines used in each round by this strategy are disjoint, and hence the final load vector is $(M, M, \ldots, M)$, with load $m^{1/p} M$; this provides an upper bound on $\OPT$.
	
	On the other hand, the algorithm considered has expected load $\E\|\sum_i \bell^i\|_p \ge \|\E \sum_i \bell^i\|_p$ (using Jensen's inequality); to simplify the notation, let $\mu = \E \sum_i \bell^i$. So taking $p$-th powers
	\begin{align}
		\bigg(\E \bigg\|\sum_{i = 1}^{\log m } \bell^i\bigg\|_p\bigg)^p \ge	\sum_j \mu_j^p \ge \sum_{i = 1}^{\log m - 1} \sum_{j \in \U_i \setminus \U_{i+1}} \mu_j^p \label{eq:LBWC}
	\end{align}
	By construction, $\U_i \setminus \U_{i+1}$ has $\frac{m}{2^i} - \frac{m}{2^{i+1}} = \frac{m}{2^{i+1}}$ machines. Moreover, for each such machine $j$, the algorithm incurs expected load at least $\frac{M}{2}$ in each of the rounds from $1, \ldots, i$, and hence $\mu_j \ge \frac{M i}{2}$. Thus the right-hand side of \eqref{eq:LBWC} is at least $$\sum_{i = 1}^{\log m - 1} \frac{m}{2^{i+1}} \left(\frac{Mi}{2}\right)^p \ge m \left(\frac{M p}{2^{2+1/p}} \right)^p,$$ where the inequality is obtained by just using the term $i = \log m - 1 = p$. Thus, the load incurred by the algorithm is at least $m^{1/p} \frac{Mp}{2^{2 + 1/p}}$, thus proving the theorem.


	\section{Proof of Theorem \ref{thm:lb}} \label{app:lb}

Let $m = 2^p$. Our instance is based on the following Walsh system. For $i = 1, \ldots, \log m$, define the vectors $v^i \in \{0,1\}^m$ as follows: construct the $m \times \log m$ 0/1 matrix $M$ by letting its rows be all the $\log m$-bit strings; then the $v^i$'s are defined as the \textbf{columns} of this matrix. We use $(v^i)^c$ to denote 0/1 vector obtained 
by flipping all the coordinates of $v^i$.

	The main motivation for this construction is the following intersection property.
	
	\begin{lemma} \label{lemma:intersection}
		Consider a subset $I \subseteq [\log m]$ and for each $i \in I$ let $u^i$ be either $v^i$ or $(v^i)^c$. Then the $u^i$'s intersect in $\frac{m}{2^{|I|}}$ coordinates, namely there is set of coordinates $J \subseteq [m]$ of size $\frac{m}{2^{|I|}}$ such that $u^i_j = 1$ for all $i \in I, j \in J$.
	\end{lemma}
	
	\begin{proof}
		Let $\bj$ be uniformly random in $[m]$. Notice that the vector $(v^1_{\bj}, v^2_{\bj}, \ldots, v^{\log m}_{\bj})$ is a random row of the matrix $M$, and hence it is a point uniformly distributed in $\{0,1\}^{\log m}$. Due to the product structure of this set, this implies that the random variables $\{v^i_{\bj}\}_i$ are all independent, and each take value $1$ with probability $1/2$. Moreover, this is true if we complement some of these variables, i.e., replace  $v^i$ for $(v^i)^c$ for some indices $i$ (notice we do now allow $v^i$ and $(v^i)^c$ to be simultaneously in the set).
		Therefore, this gives that the random variables $\{u^i_{\bj}\}_{i \in I}$ are independent, and thus the number of coordinates where they intersect is
		\begin{align*}
			m \cdot \Pr\left(\bigwedge_{i \in I} (u^i_{\bj} =  1)\right) = m \cdot \prod_{i \in I} \Pr(u^i_{\bj} =  1) = \frac{m}{2^{|I|}}.
		\end{align*}
		This concludes the proof.
	\end{proof}
	
	The instance for \GLB is then constructed randomly as follows. There are $m$ machines and $\frac{p}{2} = \frac{\log m}{2}$ jobs. 	For $i = 1, \ldots, \frac{\log m}{2}$, let $\bu^i$ be a random vector that equals $v^i$ with probability $1/2$ and equals its complement $(v^i)^c$ with probability $1/2$. Then for each $i \in [\frac{\log m}{2}]$, we have one job with only one processing option of load vector $\bu^i$, and one job with 2 processing options of load vectors $v^i$ and $(v^i)^c$. These jobs are then presented in random order. 
	
	Now we analyze this instance. The optimal offline solution can be upper bounded $\OPT \le \frac{p}{2} m^{1/p}$, since opt can process each job $\{v^i, (v^i)^c\}$ using the option that equals the complement of $\bu^i$, which gives total load vector $\frac{p}{2} \ones$, of $\ell_p$-norm $\frac{p}{2} m^{1/p}$. This is also tight, namely $\OPT \ge \frac{p}{2} m^{1/p}$: by adding up all the loads of the jobs over all the machines, we see that any solution has total $\ell_1$ load exactly $\frac{m p}{2}$, and since $\|x\|_p \ge m^{\frac{1}{p} - 1} \|x\|_1$ (see Section \ref{sec:road}) the claimed lower bound follows. 
	
	However, it is hard for the online algorithm to ``unmatch'' the processing of $\{v^i, (v^i)^c\}$ with $\bu^i$, even in the random-order model. More precisely, consider any online algorithm and let $\bX_i$ be the indicator variable that the algorithm chose to process $\{v^i, (v^i)^c\}$ using the option that \textbf{equals} $\bu^i$. Since the instance is presented in random order, with probability $1/2$ the job $\{v^i, (v^i)^c\}$ comes before the job $\bu^i$; in this case, the random variable $\bu^i$ is independent from how the algorithm processes job $\{v^i, (v^i)^c\}$, and so with probability 1/2 it equals the processing option that the algorithm chose for job $\{v^i, (v^i)^c\}$. Thus, $\E \bX_i \ge 1/4$, and hence $\E \sum_{i=1}^{p/2} \bX_i \ge \frac{p}{8}$ (i.e. on average the algorithm makes $\frac{p}{8}$ ``mistakes''). Moreover, by employing Markov's inequality to $\frac{p}{2} - \sum_i \bX_i$ we obtain that $\sum_i \bX_i \ge \frac{p}{16}$ with probability at least $1/7$.
	
	Now we see how these mistakes factor into the load of the algorithm. When $\bX_i = 0$, the processing of $\bu^i$ and $\{v^i, (v^i)^c\}$ match, adding up to load vector $\ones$, and when $\bX_i = 1$ their load vector adds up to $2 \bu^i$. Thus, the load of the algorithm is 
	\begin{align*}
		\algo = \bigg\| \sum_i (1-\bX_i) \ones + \sum_i \bX_i 2 \bu^i  \bigg\|_p.
	\end{align*}
	Using Lemma \ref{lemma:intersection} above, we see that this total load vector has at least $\frac{m}{2^{\sum_i \bX_i}}$ coordinates with value $\sum_i (1-\bX_i) + 2 \sum_i \bX_i = \frac{p}{2} + \sum_i \bX_i$. Thus,
	\begin{align*}
		\algo \ge \left(\frac{m}{2^{\sum_i \bX_i}} \left(\frac{p}{2} + \sum_i \bX_i\right)^p \right)^{1/p} = \frac{p m^{1/p}}{2} \left(1 + \frac{2 \sum_i \bX_i}{p}\right) \frac{1}{2^\frac{{\sum_i \bX_i}}{p}}.
	\end{align*}
	Notice that always $\frac{\sum_i \bX_i}{p} \le \frac{1}{2}$ and recall that with probability at least $1/7$ we have $\frac{\sum_i \bX_i}{p} \ge \frac{1}{16}$. Moreover, one can verify that the function $x \mapsto (1+2x)\frac{1}{2^x}$ is increasing in the interval $[\frac{1}{16}, \frac{1}{2}]$, so its minimum is achieved at $x=\frac{1}{16}$ with value $\ge 1.07$. Thus, with probability at least $1/7$, the algorithm incurs load at least $1.07 \frac{pm^{1/p}}{2}$. With the remaining probability, the algorithm incurs load at least that of the offline $\OPT$, which is at least $\frac{pm^{1/p}}{2}$; thus, the expected total load of the algorithm is at least $$\frac{1}{7} \frac{1.07\, pm^{1/p}}{2} + \frac{6}{7} \frac{pm^{1/p}}{2} = 1.01\, \frac{pm^{1/p}}{2}.$$ This concludes the proof of Theorem \ref{thm:lb}.


	\section{Proof of Lemma \ref{lemma:smoothened}} \label{app:smoothened}
		To simplify the notation we omit subscripts in $\psi_{\e,p}$. 
		
		\paragraph{Property \eqref{eq:propSmooth1}.} The upper bound follows directly from triangle inequality: $\psi(u) \le \frac{p}{\epsilon} (\|\mathbf{1}\|_p + \|\frac{\epsilon u}{p}\|_p) - \frac{p}{\epsilon} = \|u\|_p + \frac{p (m^{1/p} - 1)}{\epsilon}$. For the lower bound $\psi(u) \ge \|u\|_p$ for $u \in \R^m_+$, define $v = \frac{p}{\e} \frac{u}{\|u\|_p}$. Since, $\frac{p}{\e} \ones + u$ is pointwise at least $v + u$, and the later is non-negative, we have $$\bigg\|\frac{p}{\e} \ones + u\bigg\|_p\ge \|v + u\|_p = \bigg\|\left(\frac{p}{\e \|u\|_p} + 1\right) u\bigg\|_p = \|u\|_p  + \frac{p}{\e \|u\|_p} \|u\|_p \ge \|u\|_p + \frac{p}{\e}.$$ Since $\psi(u) = \|\frac{p}{\e} \ones + u\|_p - \frac{p}{\e}$ the result follows.

		%
		
		
		\paragraph{Property \eqref{eq:propSmooth2}.}
		
				Writing the partial derivatives:
		\begin{align}
			\frac{\partial \psi}{\partial x_i}(u) = \frac{\left(1 + \frac{\e u_i}{p} \right)^{p-1}}{\left( \sum_i \left(1 + \frac{\e u_i}{p}\right)^p \right)^{1-\frac{1}{p}}}. \label{eq:derPsi}
		\end{align}
		Since $u_i \ge 0$ and $v_i \in [0,1]$, we have $$\left(1 + \frac{\e u_i}{p} \right)^{p-1} \le \left(1 + \frac{\e (u_i + v_i)}{p} \right)^{p-1} \le \left(1 + \frac{\e v_i}{p} \right)^{p-1} \left(1 + \frac{\e u_i}{p} \right)^{p-1} \le e^\e \left(1 + \frac{\e u_i}{p} \right)^{p-1},$$ where in the last inequality we use the fact that for all $x \ge 0$ we have $1+x \le e^x$. Moreover, the same holds if we change the powers from $p-1$ to $p$. We can then use these inequalities in the numerator and denominator of \eqref{eq:derPsi} to obtain $\frac{\partial\psi}{\partial x_i}(u + v) \le e^\e \cdot \frac{\partial \psi}{\partial x_i}(u)$ and $\frac{\partial\psi}{\partial x_i}(u + v) \ge \frac{1}{e^{\e (1-1/p)}} \frac{\partial\psi}{\partial x_i}(u) \ge e^{-\e} \cdot \frac{\partial\psi}{\partial x_i}(u)$. This concludes the proof.


%


\section{Proof of Lemma \ref{lemma:linLp}} \label{app:proofLinLp}

		We start with the following optimal bound on the modulus of strong smoothness of the square of the $\ell_p$ norm.

		\begin{lemma}[\cite{AMS}] \label{lemma:nemirovski}
			For $p \in [2, \infty)$ and $x \in \R^m \setminus \{0\}$. Then for arbitrary $x,y \in \R^m$,
			\begin{align*}
				\|x\|_p^2 + \ip{h(x)}{y} \le \|x+y\|_p^2 \le \|x\|_p^2 + \ip{h(x)}{y} + (p-1) \|y\|_p^2,
			\end{align*}
			where the vector $h(x)$ is defined as $h(x) = 2 \|x\|_p^{2-p} \left(|x_i|^{p-2} x_i \right)_i$.
		\end{lemma}			
		
	Collecting the $\|x\|_p^2$ terms, the upper bound of this inequality becomes 
	\begin{align}
		\|x+y\|_p^2 \le \|x\|_p^2 \left(1 + \frac{\ip{h(x)}{y}}{\|x\|_p^2} + \frac{(p-1)\|y\|^2_p}{\|x\|^2_p} \right). \label{eq:nemi2}
	\end{align}
	But since $\sqrt{.}$ is concave, for any $\alpha$ we have (using linearization at $1$) $\sqrt{1+\alpha} \le 1 + \frac{\alpha}{2}$ (this can also be checked directly by squaring both sides). Thus, taking square roots on both sides of \eqref{eq:nemi2} and employing this bound we get 
	$$\|x + y\|_p \le \|x\|_p + \frac{\ip{h(x)}{y}}{2 \|x\|_p} + \frac{(p-1)}{2} \frac{\|y\|_p^2}{\|x\|_p}.$$
	
	Defining $g(x) = \frac{h(x)}{2 \|x\|_p}$ we see that this expression is exactly the one in Lemma \ref{lemma:linLp}. To conclude the proof we show that for $x \ge 0$ we have $\|g(x)\|_q \le 1$, or equivalently $\|h(x)\|_q \le 2 \|x\|_p$: noticing that $q = p/(p-1)$,
			\begin{align*}
				\|h(x)\|_q = \frac{2}{\|x\|_p^{p-2}} \left(\sum_i x_i^{(p-1) q} \right)^{1/q} = \frac{2}{\|x\|_p^{p-2}} \left(\sum_i x_i^{p} \right)^{(p-1)/p} = \frac{2 \|x\|_p^{p-1}}{\|x\|_p^{p-2}} = 2 \|x\|_p.
			\end{align*}			
			
			This concludes the proof of Lemma \ref{lemma:linLp}.


\section{Proof of Lemma \ref{lemma:refinedGuarantee}} \label{app:refinedGuarantee}
	
We reproduce the proof of \cite{caragiannis} for convenience. Let $S^t = C^1 \bar{x}^1 + \ldots + C^t \bar{x}^t$ (and $S^0 = 0$). By the greedy criterion, for each $t$, we have $\|S^t\|_p^p - \|S^{t - 1}\|_p^p \le \|S^{t -1} + C^t \bar{x}^t\|_p^p - \|S^{t - 1}\|_p^p$. Adding all these inequalities for $t \ge \tau$ we get
	\begin{align}
 		\|S^n\|_p^p - \|S^{\tau-1}\|_p^p \le \sum_{t \ge \tau} \left(\|S^{t -1} + C^t \bar{x}^t\|_p^p - \|S^{t - 1}\|_p^p\right) \le \sum_{t \ge \tau}\left(\|S^n + C^t \bar{x}^t\|_p^p - \|S^n\|_p^p\right), \label{eq:refGuarantee1}
 	\end{align}
 	 where the last inequality follows from the fact that $x \mapsto (x + a)^p - x^p$ is non-decreasing over $[0, \infty)$ for $a \ge 0$.
	Employing Lemma 3.1 of \cite{caragiannis} 
	we get that the right-hand side is at most
	\begin{align}
		\textrm{RHS} \le \bigg\|S^n + \sum_{t \ge \tau} C^t \bar{x}^t\bigg\|_p^p - \|S^n\|_p^p \le \left(\|S^n\|_p + \bigg\|\sum_{t \ge \tau} C^t \bar{x}^t\bigg\|_p\right)^p - \|S^n\|_p^p . \label{eq:refGuarantee2}
	\end{align}
%

	If $\|S^n\|_p^p \le 2 \|S^{\tau-1}\|_p^p$ then Lemma \ref{lemma:refinedGuarantee} clearly holds. Otherwise $\|S^{\tau-1}\|_p^p \le \frac{\|S^n\|_p^p}{2}$, which used together with inequalities \eqref{eq:refGuarantee1} and \eqref{eq:refGuarantee2} gives $$\frac{3}{2} \|S^n\|_p^p \le \left(\|S^n\|_p + \bigg\|\sum_{t \ge \tau} C^t \bar{x}^t\bigg\|_p\right)^p.$$ Taking $p$-th roots and reorganizing we get $\|S^n\|_p \le \frac{1}{(\frac{3}{2})^{1/p} - 1} \|\sum_{t \ge \tau} C^t \bar{x}^t\|_p$. Using the inequality $e^x \ge 1 + x$, we have that $\frac{1}{(\frac{3}{2})^{1/p} - 1} \le \frac{p}{\ln (3/2)}$, which gives the desired result. This concludes the proof of the lemma.


	\section{Proof of Lemma \ref{lemma:greedyGrad}} \label{app:greedyGrad}
	
    	Let $\psi := \psi_{\e,p}$ to simplify the notation. Integrating and using Property \eqref{eq:propSmooth2} we get $$\psi(u + v) = \psi(u) + \int_0^1 \ip{\nabla \psi(u + \lambda v)}{v} d\lambda ~\in~ \psi(u) + e^{\pm \e} \ip{\nabla \psi(u)}{v},$$ and similarly for $v'$. Thus,
    	\begin{align*}
    		\ip{\nabla \psi(u)}{v} \le e^\e \left[\psi(u + v) - \psi(u)\right] \le e^\e \left[\psi(u + v') - \psi(u)\right] \le e^{2\e} \ip{\nabla \psi(u)}{v'}.
    	\end{align*}	
		This concludes the proof. 
		
\end{document}